\documentclass{article}

\usepackage[nonatbib,preprint]{neurips_2024}
\usepackage[pdftex]{graphicx}
\usepackage{algorithm}
\usepackage{algpseudocode}
\RequirePackage{doi}
\usepackage{hyperref}
\usepackage{subcaption}
\usepackage{multirow}

\usepackage[utf8]{inputenc} 
\usepackage[T1]{fontenc}    
\usepackage{hyperref}       
\usepackage{url}            
\usepackage{booktabs}       
\usepackage{amsfonts}       
\usepackage{nicefrac}       
\usepackage{microtype}      
\usepackage{xcolor}         
\usepackage{amsmath}
\usepackage{amsthm}
\usepackage{wrapfig}

\newcommand*{\N}{\mathbb N}
\newcommand*{\R}{\mathbb R}
\renewcommand*{\P}{\mathbb P}
\newcommand*{\E}{\mathbb E}
\newcommand*{\Z}{\mathbb Z}

\newcommand*{\kA}{\mathcal A}
\newcommand*{\kS}{\mathcal S}

\newcommand*{\define}{\overset{\text{Def.}}=}

\newtheorem{lemma}{Lemma}

\theoremstyle{definition}
\newtheorem{definition}{Definition}

\theoremstyle{remark}
\newtheorem{remark}{Remark}

\title{Human-in-the-loop Reinforcement Learning for Data~Quality~Monitoring in Particle~Physics~Experiments}

\author{%
  Olivia Jullian Parra \\
  Department of Experimental Physics\\
  European Organization \\
  for Nuclear Research (CERN) \\
  1211 Geneva 23, Switzerland \\
  \texttt{olivia.jullian.parra@cern.ch} \\
  \And
  Julián García Pardiñas \\
  Department of Experimental Physics\\
  European Organization \\
  for Nuclear Research (CERN) \\
  1211 Geneva 23, Switzerland \\
  \texttt{julian.garcia.pardinas@cern.ch} \\
  \And
  Lorenzo Del Pianta Pérez \\
  Department of International Relations\\
  European Organization \\
  for Nuclear Research (CERN) \\
  1211 Geneva 23, Switzerland \\
  \texttt{lorenzo.del.pianta.perez@cern.ch} \\
  \And
  Maximilian Janisch \\
  Department of Mathematics\\
  University of Zürich\\
  Winterthurerstrasse 190, \\
  Zürich 8057, Switzerland \\
  \texttt{maximilian.janisch@math.uzh.ch} \\
  \And
  Suzanne Klaver \\
  Faculteit der B\`etawetenschappen\\
  Vrije Universiteit Amsterdam\\
  De Boelelaan 1105, 1081 HV \\
  Amsterdam, the Netherlands \\
  \texttt{suzanne.klaver@cern.ch} \\
  \And
  Thomas Lehéricy \\
  Department of Mathematics\\
  University of Zürich\\
  Winterthurerstrasse 190, \\
  Zürich 8057, Switzerland \\
  \texttt{thomas.lehericy@math.uzh.ch} \\
  \And
  Nicola Serra \\
  Department of Physics\\
  University of Zürich \\
  Winterthurerstrasse 190, Zürich 8057, Switzerland \\
  \texttt{nicola.serra@cern.ch} \\
}

\begin{document}

\maketitle

\vspace{-5pt}

\begin{abstract}
Data Quality Monitoring (DQM) is a crucial task in large particle physics experiments, since detector malfunctioning can compromise the data. DQM is currently performed by human shifters, which is costly and results in limited accuracy. In this work, we provide a proof-of-concept for applying human-in-the-loop Reinforcement Learning (RL) to automate the DQM process while adapting to operating conditions that change over time. We implement a prototype based on the Proximal Policy Optimization (PPO) algorithm and validate it on a simplified synthetic dataset. We demonstrate how a multi-agent system can be trained for continuous automated monitoring during data collection, with human intervention actively requested only when relevant. We show that random, unbiased noise in human classification can be reduced, leading to an improved accuracy over the baseline. Additionally, we propose data augmentation techniques to deal with scarce data and to accelerate the learning process. Finally, we discuss further steps needed to implement the approach in the real world, including protocols for periodic control of the algorithm's outputs.
\end{abstract}

\section{Introduction}
\label{sec:introduction}

Large particle physics experiments, such as those conducted at the Large Hadron Collider~\cite{ATLAS:2008xda,CMS:2008xjf,LHCb:2008vvz,ALICE:2008ngc}, produce data for physics studies by colliding particles continuously over extended periods of time. The total amount of data that can be collected is limited by the lifetime of the experiment. The data is recorded by complex detectors that can malfunction with an unknown frequency and in a variable and unpredictable manner. Unless promptly identified and fixed, these issues can render the data unusable for scientific studies. Consequently, Data Quality Monitoring (DQM) procedures that periodically inspect the data while being collected are of utmost importance.

The task of DQM in particle experiments is typically divided into two subsequent stages, one online and one offline. The online stage consists on the continuous monitoring of the data synchronously with the data collection process. This is typically done at LHC experiments at fixed time intervals of several minutes (time needed to accumulate enough counts for a meaningful statistical evaluation). When anomalies are detected, they are fixed before continuing the data collection. The main goal of this DQM stage is then to inspect the data to detect anomalies as fast as possible, ideally every time the histograms are re-filled. After the data has been recorded, it is inspected a second time in the offline stage. The lack of intrinsic time constraints allows to evaluate high-level distributions that require larger data sets and more time to reconstruct, to spot anomalies that have not been found in the online step and consequently discard the corresponding data. The main goal of this DQM stage is a high accuracy in the classification between good and anomalous data. 

Both DQM stages are typically performed visually by a pool of rotating non-expert shifters, that need to examine a large set of histograms representing different data features (for example information coming from different subdetectors) at every monitoring step. This manual approach presents multiple challenges. First of all, it is highly costly in terms of human resources, as each experiment requires hundreds of shifters per year for the task. Second, the shifters can have different criteria when judging the data quality and different levels of specific background knowledge, which introduces noise in the classification and hence reduces the accuracy. Finally, the human attention is naturally limited: the shifters are typically not able to thoroughly inspect all the histograms for all the monitoring steps (reducing the classification accuracy) and not able to keep up with the full rate of monitoring steps in the online regime (reducing the speed in the detection of anomalies).

The previous challenges can be overcome by adding automation to the task, and this is progressively done in multiple experiments, with machine learning (ML) techniques being particularly effective. One typical challenge that many of those approaches do not consider is the adaptation to operational conditions that change over time. For example, the particle beam conditions can be changed at different points in time, or the different subdetectors of a given experiment can experience multiple tunings for a long period following their first commissioning. These changes in conditions prevent the usage of ML techniques that are trained in a stationary regime, unless they are being manually re-trained every time it is needed.

This letter proposes for the first time the use of Reinforcement Learning (RL)~\cite{Sutton1998} to perform DQM at particle physics experiments, which offers a native way to adapt to changing conditions while capturing the interdependencies between the different data features. As a second benefit, RL techniques bring the unique possibility to automate more complex actions than a simple classification of good and anomalous data, with huge potential to further reduce person power costs while increasing the overall efficiency of the operations at an experiment's control room. While that second advantage refers to the online regime, we propose and study RL approaches both for online and offline, given the similarities between both regimes and the benefits of having common software approaches within a given experiment, in terms of development and maintainability. We propose to train the algorithms purely from the shifter's feedback, which is an example of Reinforcement Learning from Human Feedback (RLHF)~\cite{kaufmann2023survey}. This choice prevents the time-costly task of manually defining and updating histogram references for all the data features. We thus leverage the domain knowledge and intuition of the shifters, combined with up-to-date instructions from the detector experts. Finally, we propose to increase the sample efficiency of the training by providing approaches to augment the data, creating artificial data points with and without anomalies.

This letter illustrates the novel approach through an example implementation strategy for the online and offline use cases. The feasibility of the strategy is demonstrated through a series of proof-of-concept studies based on a simplified synthetic dataset. Regarding the RL algorithm, an actor-critic model~\cite{Sutton1998} is used together with a Proximal Policy Optimization (PPO)~\cite{schulman2017proximal} loss. This architecture is state-of-the-art and presents good properties in terms of training stability. A multi-agent~\cite{marl-book} structure is used for the online regime, while the offline regime uses a single agent. It should be noted that, since the goal of this work is to provide a proof of concept and not achieve a certain level of performance, the hyperparameters of the different setups have not been optimised beyond few iterations of manual tuning in each case.

The structure of the present paper is as follows. We start by reviewing the current state of the literature in Section~\ref{sec:related}. Then, the specific setup used for the proof-of-concept experiments is described in Section~\ref{sec:setup} and the technical implementation of the RL algorithms is presented in Section~\ref{sec:algorithm}. The experiments done in the offline and online regimes are described in Section~\ref{sec:offline} and Section~\ref{sec:online}, respectively. It should be noted that, while the online stage precedes the offline one in the real application, they are studied in inverse order in this work, since the offline stage represents a simpler case and serves as a basis for the online studies. The results of the experiments, their limitations and possible directions of future work are discussed in Section~\ref{sec:discussion}. Finally, the main conclusions of the paper are highlighted in Section~\ref{sec:conclusions}
\section{Related work}
\label{sec:related}

Although the DQM task in particle physics experiments is predominantly manual, there have been works that make use of machine learning techniques (see e.g. Ref.~\cite{doi:10.1142/9789811234033_0005} for a review). While some of those approaches are based on supervised learning (e.g. using Boosted Decision Trees~\cite{Adinolfi:2298467} or Convolutional Neural Networks~\cite{Pol:2683825}), most of them are based on semi-supervised learning (e.g. using Autoencoders~\cite{Pol:2650715, CMSHCAL:2023skb, Deja:2707754} or Variational Autoencoders~\cite{pol:hal-02428005}).

This letter proposes for the first time the use of RL techniques for the task of DQM in particle physics experiments. Those types of techniques have already been used for anomaly detection tasks in a broad variety of other systems (see Ref.~\cite{9956995} for a recent review), bringing as main advantages the capacity to adapt to changing conditions and to deal with highly-complex systems. Concerning the field of particle physics, RL has been used for tasks other than DQM, mostly connected to particle accelerator control~\cite{Hanten:2019khc, Xu:2023erw, Hirlaender:2023ivj, Assmann:2023vkq, Kafkes:2021jse, Hirlaender:2020eky, Kain:2020vjs, Pang:2020ose}. The efforts in that direction show great promise, although identify the low sample efficiency of RL algorithms as a problem to do a continuous training during system operation. This type of problem can also appear for the task of DQM discussed in this paper. To tackle it, we propose the usage of data-augmentation techniques, that have been broadly used to improve the training of deep learning models in multiple contexts. As an example, \cite{ijcai2021p631} and \cite{Shorten2019ASO} review data augmentation techniques used in the time series and image domains, respectively.

From the different types of RL, the work in this letter focuses on RLHL, characterised by the inclusion of human feedback in the reward function. This modification allows to align better the algorithm's behaviour to the human goals and is particularly useful when those goals are difficult to specify but easy for a human to judge. A recent review on types and applications of RLHL can be found in \cite{kaufmann2023survey}. The most typical applications, described in that review, are control systems and robotics, Large Language Models (a prime example of which is ChatGPT~\cite{ChatGPT}) and other generative tasks, and recommender systems.
\section{Experimental setup}
\label{sec:setup}

A machine starts in a nominal mode of operation, and is subject to random failures which shift it to one of several anomalous modes of operation. The machine produces a summary statistic of its mode of operation at regular intervals --- a \emph{data point}, in this paper a one-dimensional histogram represented by a fixed-length vector of real values. This constitutes a simplification of a typical real application, where each data point is represented by many different histograms. For each data point, we may have access to the machine’s mode of operation via a label provided by a shifter, indicating whether the mode of operation is nominal or anomalous. Depending on the experiment, the labels can reflect the ground truth or include some level of shifter-related noise. Since we aim to provide a proof-of-concept, we exclusively use synthetic data. Its generation is described in \ref{sec:dataset}. 

We distinguish two training regimes, \emph{online} and \emph{offline}. 
The labels will always be available in the offline case, and in the online case, only when requested by the algorithm. The transition between the modes of operation of the machine, which constitutes a Markovian process, is different between the offline and online regimes, and it will be specified in the following sub-sections. 
The Reinforcement Learning point of view is elucidated further in Section \ref{sect:environment}. The data points follow the same order as they would in real-world collection. This allows for studying the capability of our RL algorithm to adapt to changing conditions, especially abrupt changes to the distribution of anomalous and nominal data points.

\subsection{Offline regime}
In the offline regime, shifters will inspect all the histograms, following the order in which they have been collected, and classify them as either nominal or anomalous. The machine's mode of operation is reset to operational after every time step, at which it then breaks with a fixed probability $p_{\text{anom}}$ (in applications, this is done by censoring the sequence of observed data points), so that its modes of operation are independent and identically distributed. In our experiments, we choose $p_{\text{anom}}=0.3$. The distribution of the data points may still change over time due to changes of the data point distribution conditional on the mode of operation of the machine. 

\subsection{Online regime}

In this regime, the data is collected while the algorithm is running. This regime introduces an additional level of complexity with respect to the offline one: the need to repair the machine. Namely, if an anomaly is spotted (either by the shifter or by the algorithm), the machine is restarted and returned to its nominal mode of operation. Otherwise, the mode of operation of the machine remains unchanged. In both cases, the machine then evolves to a new mode of operation: if it is in an anomalous mode of operation, its mode of operation remains the same, and if the machine is in the nominal mode of operation, the next mode of operation will be anomalous with probability $p_{anom}$ (set to the same value as in the offline regime) or nominal with probability $1-p_{anom}$.

As a second additional level of complexity, the shifter's labels used for the training are not always accessible, to emulate the time limitations experienced by the shifters working in the online regime. This will be discussed in detail in Sec.~\ref{sec:algorithm}.

\subsection{Computing and software resources}
\label{sec:resources}

The experiments described in this paper have been performed in a commercial MacBook Pro laptop, with an Apple M1 Pro chip, 8 CPU cores and 16 GB of RAM. All the experiments are lightweight, presenting minimal resources in terms of computing time (order of few minutes per experiment, except for the one described in Sec.~\ref{sec:exp_data_augmentation}, that takes about 4 hours) and memory consumption. The algorithm is written in PyTorch~\cite{DBLP:journals/corr/abs-1912-01703}, adapting the PPO implementation in Ref.~\cite{pytorch_minimal_ppo}. The Gym library~\cite{DBLP:journals/corr/BrockmanCPSSTZ16} is used to create the RL environment.
\section{Reinforcement Learning algorithm}
\label{sec:algorithm}
We propose a multi-agent Proximal Policy Optimization (PPO) algorithm with two agents: A \emph{predictor}, whose task is to assign a label (nominal / anomalous) to each histogram; and a \emph{checker}, whose task is to decide, based on the histogram and on the decision of the predictor, whether to call a human shifter who will themselves provide a label. 
Readers unfamiliar with PPO will find a conceptual summary of the approach in appendix \ref{sect:appendix-PPO}.

We describe the environment in section \ref{sect:environment}, the actions and value functions of the actors in section \ref{sect:actions}, the training episodes in section \ref{sect:episodes}, the rewards in section \ref{sect:reward}, the loss function optimized by the PPO approach in section \ref{sect:loss}, and finally the training details in section \ref{sect:training}.

\subsection{Environment}\label{sect:environment}
At a given time $t$, following the description in section \ref{sec:setup}, the environment state 
consists of the current data point and of the current mode of operation (nominal / anomalous) of the machine. 
If a shifter is called or if we are in the offline regime, the state 
may be augmented by a human label. Only the data point is available to the predictor and to the checker.

\subsection{Agents and actions}\label{sect:actions}
There are two agents in our approach: The \emph{predictor} and the \emph{checker}. The predictor is in charge of deciding whether the mode of operation in a given state is nominal or anomalous --- recall that only the data point is available to the predictor. Which decision to take defines the two actions of the predictor. The checker has to examine the predictor's response for a given data point and decide whether to obtain a label from an external shifter or not. If called by the checker, a shifter will provide a label that can be used to train the predictor. If not called by the checker, in the online regime, there is a fixed probability $p$ (that we set to 0.1) with which a shifter will provide a label anyway. In the offline case, a label will always be provided, meaning that the checker’s actions have no effect on the state evolution and can thus be discarded.  Since we are using PPO, we need to train, for each agent, two networks simultaneously (conceptual details can be found in appendix \ref{sect:PPO-final}): The \emph{actor}, defining the policy of the agent, and the \emph{critic}, approximating the value function of the decision process. In our implementation, the predictor's and checker's actor and critic are all provided by a MultiLayer Perceptron (MLP) with one hidden layer of 10 units using the rectified linear unit (ReLU) activation function. In total we thus use $4$ MLPs in the online regime and $2$ in the offline regime.

\subsection{Training episodes}\label{sect:episodes}
In the offline regime, each individual histogram constitutes one training episode. In the online regime, all histograms between one random check, as described in \ref{sect:actions}, and the next one, define the training episode. This is independent of the actions of the checker (if the checker requests a check, an episode divide happens if and only if a random check would have occurred without the request of the checker). It should be noted that this definition of episode is merely a design choice and is not expected to significantly impact the results compared to other possible approaches to the same problem.

\subsection{Rewards} \label{sect:reward}
In this section, we describe the rewards for the predictor and checker agents. A conceptual overview of Reinforcement Learning rewards can be found in appendix \ref{sect:formal-setting}.

\paragraph{Predictor} If the shifter's label is available in the current state, the reward is $+1$ if the predictor’s decision matches the human label, and $-1$ otherwise. If the shifter's label is not available, the predictor's reward is zero.

\paragraph{Checker} If the checker did not request a check, the reward is zero. If the checker requested a check, the reward is $\omega-p$, where $\omega$ is for the predictor's \emph{mis-tagging ``probability’’} (defined below) and $p \in (0,1)$ is a hyperparameter (we set it to $0.1$) that regulates the amount of penalisation given to the checker for calling the shifter ``unnecessarily’’, i.e. when the predictor is doing well. The $\omega$ variable is a proxy of the probability that the predictor outputs the wrong label, using its own output. It is constructed using the two logits $(lp_\text{n}, lp_\text{a})$ in the predictor's output, for nominal and anomalous predictions respectively. We compute ``probabilities’’ by passing this vector through a softmax layer, then define $\omega$ as the ``probability’’ associated to the outcome that was not chosen by the shifter, i.e. $\omega = p_\text{a}$ if the shifter’s label is ``nominal’’, and $\omega = p_\text{n}$ otherwise. The checker is thus positively rewarded if it calls the shifter when the predictor's mis-tagging ``probability’’ for a given state is larger than $p$, and negatively rewarded otherwise.  It should be noted that, with this reward scheme, the checker could in principle adopt a policy in which it never asks for a check. Since all rewards would then be $0$, this policy would be stable. However, this collapse is prevented by the continuous exploration induced by the entropy term included in the PPO loss, as described in section~\ref{sect:loss}.

\subsection{Losses}\label{sect:loss}

The loss $L$ we use is a modification of the standard clipped PPO loss $L^{\text{clip}}$ described in appendix \ref{sect:PPO-final}. For a given agent, it consists of three terms, as in \cite[Equation (9)]{schulman2017proximal},
$L = L^{\text{clip}} - c_1 L^{\text{value}} + c_2 L^{\text{entropy}}$,
where the hyper-parameters $c_1$ and $c_2$ are non-negative real numbers. We use $c_1 = \frac 12$ and $c_2 = \frac 1{100}$ unless otherwise stated. The first term is simply the clipped loss used to update the parameters of the actor. The $L^{\text{value}}$ loss is the loss for the critics, equal to the squared difference between the prediction of the value network and the observed reward, $L^\text{value} = (V^t_\phi(S_t)-R_t)^2$, see appendix \ref{sect:PPO-final} for details. Finally, $L^{\text{entropy}}$ is the entropy of the policy on the current data point. This term is added to encourage exploration that may lead out of local minima. 

\subsection{Network update}\label{sect:training}

The procedure to update the actor and critic networks (see appendix \ref{sect:PPO-final}) makes use of the Adam optimizer \cite{Adam} with a learning rate of $10^{-4}$. A sequence of 10 update steps is done for each state.

\subsection{Training and testing scheme}\label{sect:training-testing}

The training of the algorithms in the different experiments will be done following the order of the sequence of data points, training on a single data point at a time and over the whole dataset in the experiment. The accuracy of the algorithm is evaluated on the ``future'' episode immediately after the current one. 
Unless explicitly stated, all the plots presented in the following sections use this method to compute accuracy.

To evaluate how the performance depends on the specific training trajectory, each experiment is repeated a certain number of times (10, unless otherwise specified), changing the seed of the random number generator used in the generation of the states and in the initialisation of the weights. As a result of the procedure, the plots shown in the following sections contain solid lines and error bands representing, respectively, the average and standard deviation over experiments. A binning of 20 episodes is used when computing those statistics in all the plots except for the ones in Section~\ref{sec:exp_data_augmentation}, where they are computed for single episodes.
\section{Experiments in the offline regime}
\label{sec:offline}

In the following subsections, we study the algorithm's capacity to adapt to changing conditions (Section~\ref{sec:exp_changing_conditions}) and to improve the classification accuracy over a noisy shifter's baseline. To study the improvement over a noisy baseline, we first investigate a simpler setup, in which the shifter's labels are independent of the algorithm's predictions (Section~\ref{sec:exp_superhuman}). We then explore the case in which the shifter randomly decides to follow the decision of the algorithm, with the probability of this decision being dependent on some estimate of the predictor's accuracy (Section~\ref{sec:exp_superhuman_with_human}). Finally, we study the improvements brought to the training by a data-augmentation technique (Section~\ref{sec:exp_data_augmentation}).

\subsection{Capacity of the algorithm to adapt to changing conditions}
\label{sec:exp_changing_conditions}

For this experiment, we generated a sequence of 2000 states. The type of nominal distribution is changed from the 1000\textsuperscript{th} state onwards (the specific distributions and parameters used are summarized in Table~\ref{tab:distribution_params_setup1} of Appendix~\ref{sec:dataset}). The evolution of the accuracy over the sequence of states (computed as explained in Section~\ref{sect:episodes}) is shown in Figure~\ref{fig:changing_accu}. As expected, the accuracy first increases as the algorithm learns to distinguish the initial set of nominal and anomalous distributions. The accuracy drops at state number 1000 due to the change in conditions, and finally recovers after adapting to the new situation.

\begin{figure}
  \centering
  \includegraphics[width=0.5\linewidth]{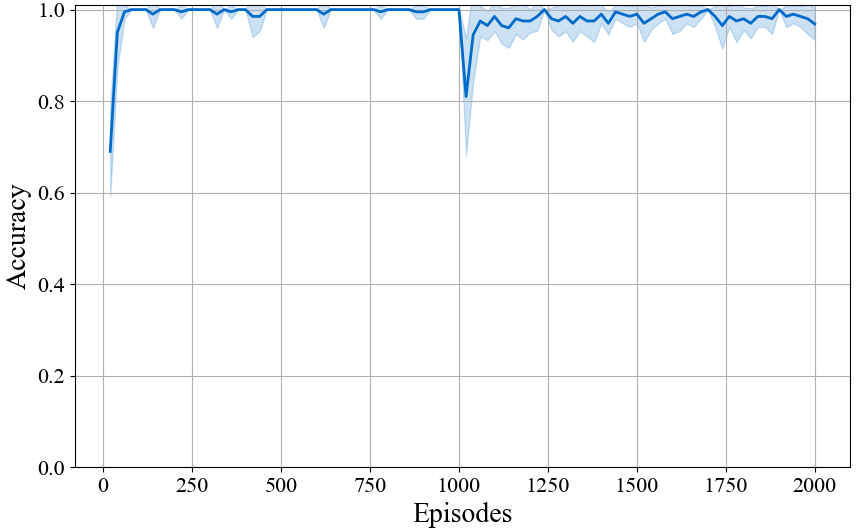}
  \caption{Results of the experiment in which the capacity to adapt to changing conditions is tested.%
  }
  \label{fig:changing_accu}
 \end{figure} 

\subsection{Ability to learn from noisy labels}
\label{sec:exp_superhuman}

In this experiment, we generate noisy shifter labels by randomly switching the ground truth of mode of operation with a probability of $30\%$, to emulate a noisy baseline accuracy of $70\%$. We generate a sequence of 2000 states with time-independent nominal and anomalous distributions (details in Table~\ref{tab:distribution_params_setup2} of Appendix~\ref{sec:dataset}) and evaluate the evolution of the algorithm's training using these noisy labels. The resulting accuracy of the algorithm's labels, judged against the ground truth of the modes of operation, is shown in Figure~\ref{fig:superhuman_condition}. It is compared to the baseline accuracy of the training labels. The algorithm's accuracy surpasses the baseline, demonstrating the ability of the algorithm to learn signal from noisy labels without learning the noise, a classically observed capability of machine learning algorithms \cite{conf/nips/NatarajanDRT13}.

\subsection{Human-machine mutual feedback}
\label{sec:exp_superhuman_with_human}
The previous experiment demonstrates the ability of our algorithm to improve over a noisy baseline, as is expected based on the machine learning literature \cite{conf/nips/NatarajanDRT13}. We now investigate what happens if the training data, i.e. the labels of the shifters, are influenced by previous outputs of the algorithm. In the real world, this will happen because shifters will have access to the output of the algorithm when making their decision.

For a simplified demonstration, we generate a sequence of 2000 states under the same conditions as in Section~\ref{sec:exp_superhuman}, except that we modify the decision-making of the shifter as follows. First, the labels of the shifter are a noisy version of the ground truth, with $30\%$ of labels randomly changed, as in Section \ref{sec:exp_superhuman}. However, the shifter now gets access to the label provided by the algorithm, as well as a proxy for the probability with which the algorithm thinks it is correct, $p^{\text{proxy}}_{\text{correct}}$, here taken to be the softmax of the predictor's output logit with the highest value (see Section~\ref{sect:reward} for a similar discussion). The shifter now randomly decides to follow the current output of the algorithm, with a probability that depends on $p^{\text{proxy}}_{\text{correct}}$ (details in Appendix \ref{sec:trust_function}). The predictor is trained with these changed shifter labels. In a further line of research, we may aim to make the model more realistic by having the shifters confidence in their decision also influence the final decision, possibly by having the shifter perform Bayesian inference.

If the shifter were to never change their decision, this experiment would be identical to the one in Section \ref{sec:exp_superhuman}. If the shifter always followed the decision of the algorithm, there would be no further training of the algorithm.

The results are presented in Figure~\ref{fig:super_human_with_trust}. They demonstrate that in the current setup, the algorithm still trains successfully and is able to improve over the noisy baseline, despite a part of its training data having been generated by a previous version of the algorithm.

\begin{figure}
    \centering
    \begin{subfigure}[t]{0.5\textwidth}
        \centering
        \includegraphics[width=0.99\linewidth]{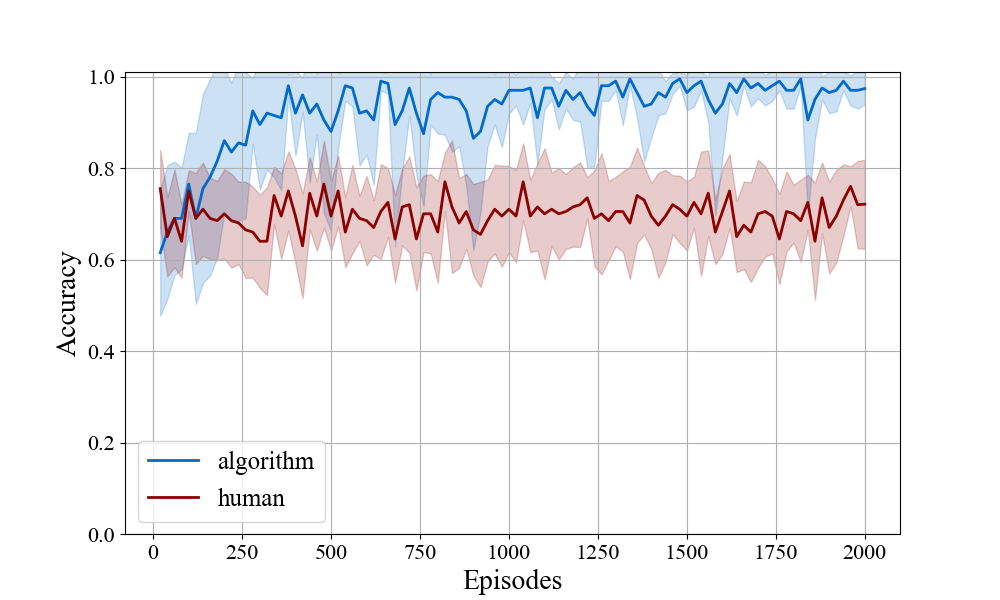}
        \caption{Shifter decisions before seeing algorithm outputs.}
        \label{fig:superhuman_condition}
    \end{subfigure}%
    ~ 
    \begin{subfigure}[t]{0.5\textwidth}
        \centering
        \includegraphics[width=0.99\linewidth]{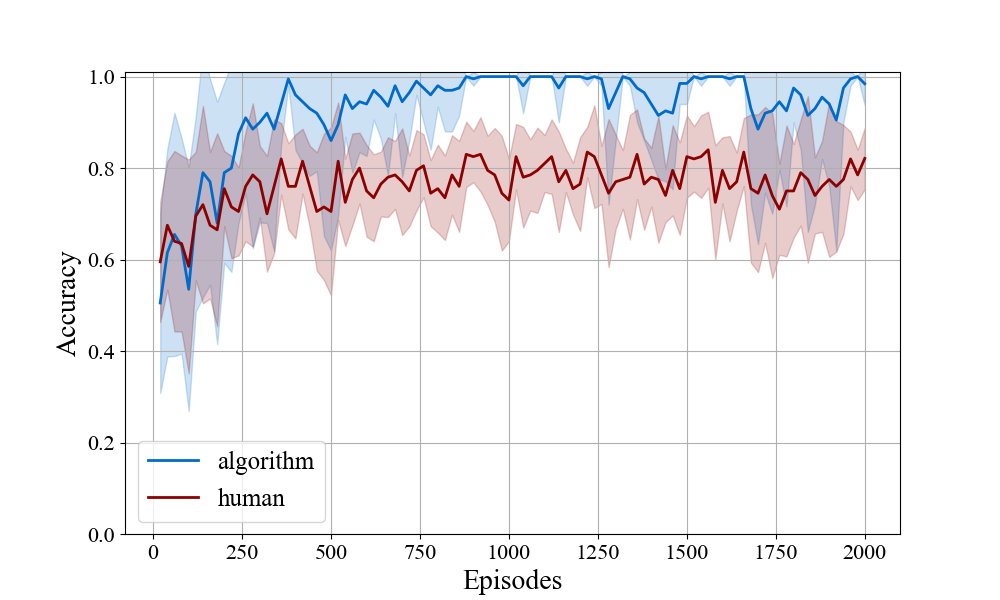}
        \caption{Shifter decisions after seeing algorithm outputs.}
        \label{fig:super_human_with_trust}
    \end{subfigure}
    \caption{Results of the experiments in the offline regime with noisy shifter labels.}
\end{figure}

\subsection{Data augmentation}
\label{sec:exp_data_augmentation}

In real situations, the training data is scarce. This causes two main problems. First, if the dataset is very imbalanced and only a small number of anomalous examples is available, the algorithm may overfit on them and may fail to generalize to other types of anomalies. Second, the number of iterations required by the algorithm to adapt to new conditions combined with the fact that the data points are collected at fixed time intervals translates into an effective time window that may be too large in cases where quick adaptation is required. 

To tackle both challenges, we propose a data-augmentation approach wherein we insert a fixed number of artificially generated data points right after each real data point. Our algorithm is described in Appendix~\ref{sec:augmentation}: it generates histograms based on an automatically built reference for the nominal state and on a large set of predefined types of anomalous deviations. The hyperparameters that control the data augmentation process and their chosen values for the current setup are gathered in Table~\ref{tab:data_augmentation_params}.

\begin{wrapfigure}{r}{0.5\textwidth}
  \centering
  \includegraphics[width=\linewidth]{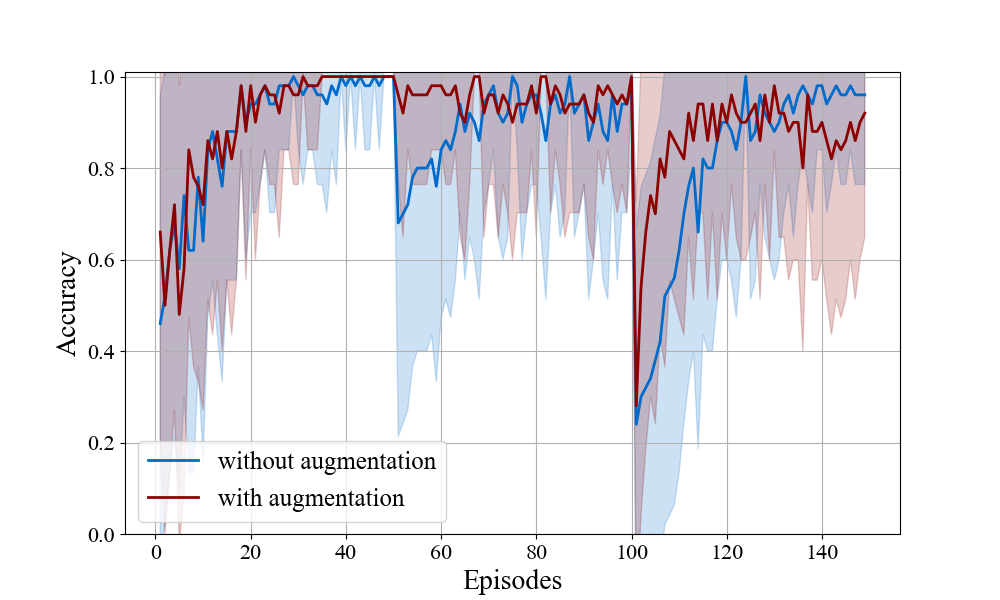}
  \caption{Accuracy (computed on the non-augmented dataset) of the algorithm trained on the dataset augmented with our approach, resp. not augmented.}
  \label{fig:data_augmentation}
\end{wrapfigure}

To test the procedure, we generate a sequence of 150 episodes, changing the type of anomalous distribution after the first 50 episodes and then the type of nominal distribution after the first 100 episodes. The labels provided by the shifter match the ground truth. The specific distributions and parameters used are gathered in Table~\ref{tab:distribution_params_setup3}. We compare the training in two different cases, with and without data augmentation. The experiments are run 50 times instead of 10 to allow for a better comparison of the shape of the two learning curves. In addition, the entropy coefficient (see Section \ref{sect:loss}) is changed from 0.01 to 1, to encourage exploration during training. We implemented this change after observing a collapse of the policy when abruptly changing nominal conditions with data augmentation; we deem that this collapse originates from a policy overspecialisation in the augmented dataset.

The results are shown in Figure~\ref{fig:data_augmentation}. We first observe that the data augmentation prevents the drop in the accuracy after 50 episodes, showing that it helps the training to generalise to different types of anomalies. Secondly, the learning curve after 100 episodes is visibly steeper when data augmentation is used than when it is not, showing that our data augmentation process can speed up the adaptation of the algorithm to new conditions. It should be noted however that the addition of artificial anomalies obliges the algorithm to solve a more complicated problem than the original one. This would tend to increase the learning time, going against the original goal, but the situation can be solved by an adequate hyper-parameter tuning, as in the case shown here.

\section{Experiments in the online regime}
\label{sec:online}

We generate 2000 episodes, changing the conditions after the first 1000. The distributions and parameters used (same as in Section~\ref{sec:exp_changing_conditions}) can be found in Table~\ref{tab:distribution_params_setup1} of Appendix~\ref{sec:dataset}. The shifter’s labels, when provided, match the ground truth. 
The estimated accuracy for the predictor is shown in Figure~\ref{fig:online_predictor}, while the fraction of times the checker (estimated in a similar manner as the accuracy) calls the agent is shown in Figure~\ref{fig:online_checker}. 

The predictor's accuracy follows the same type of behaviour as shown in Figure~\ref{fig:changing_accu}, which demonstrates that it trains well and adapts successfully to changing conditions. We observe that the checker calls the shifter mostly at two points in time: at the beginning, when the predictor is not yet trained and its accuracy is consequently poor, and right after the change in conditions, when the predictor's accuracy drops until it learns again. If the accuracy of the predictor is high, the shifter is rarely called. 

\begin{figure}
    \centering
    \begin{subfigure}[t]{0.5\textwidth}
        \centering
        \includegraphics[width=0.95\linewidth]{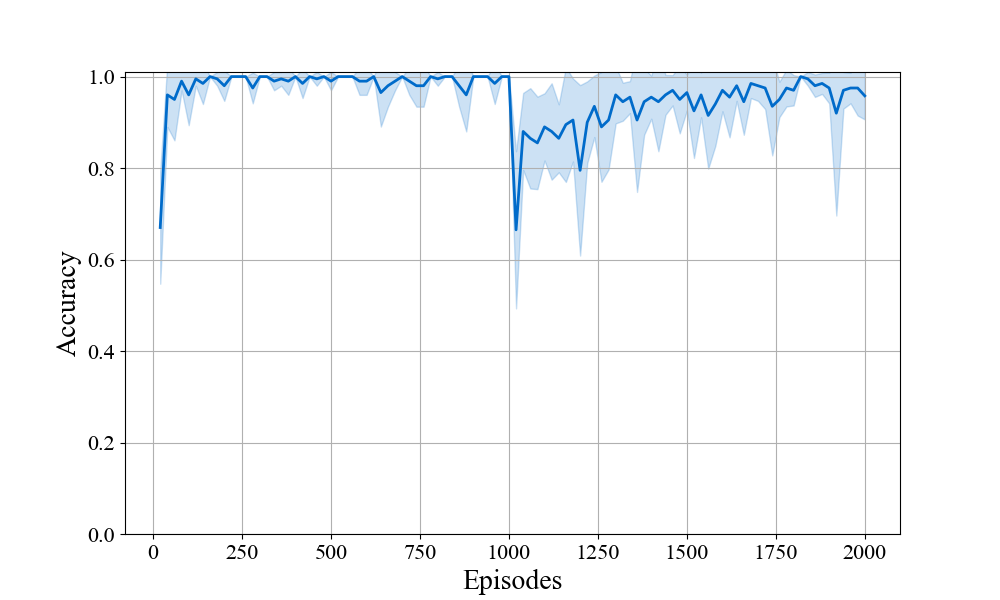}
        \caption{Predictor's accuracy.}
        \label{fig:online_predictor}
    \end{subfigure}%
    ~ 
    \begin{subfigure}[t]{0.5\textwidth}
        \centering
        \includegraphics[width=0.95\linewidth]{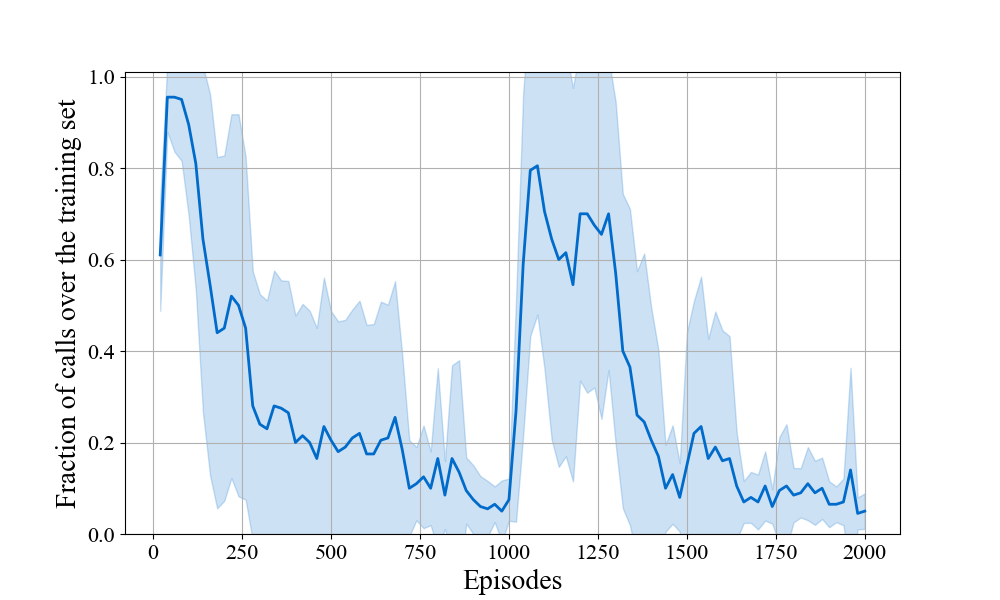}
        \caption{Checker's actions, evaluated on the training steps.}
        \label{fig:online_checker}
    \end{subfigure}
    \caption{Results of the experiment in the online regime.}
\end{figure}

\section{Discussion}
\label{sec:discussion}

\subsection{General discussion}

The experiment presented in Section~\ref{sec:exp_superhuman_with_human} shows how the RLHF approach can be used to perform a continuous training of both algorithm and shifters based on mutual feedback, resulting in an increased accuracy of the data classification compared to a shifter-only baseline. In that section, we describe a setup in which the shifter can judge how much to trust the algorithm's prediction based on the logits of its output for the current state. An alternative approach that is expected to give similar results would consist on keeping a record of the average accuracy achieved by the algorithm in the recent past (measured against the past shifter's labels), and use that as the $p_{\text{corr}}^{\text{proxy}}$ variable.

Section~\ref{sec:exp_changing_conditions} shows that the algorithm can adapt to changes in the function of the machine. We expect that in practical cases, with a careful choice of hyperparameters, the obtained accuracy will be at least as good as the one obtained if training the algorithm from scratch.
In addition, Section~\ref{sec:exp_data_augmentation} shows that the speed of that adaptation can be improved with data-augmentation techniques. Those types of techniques serve the additional purpose of compensating for low sample efficiency in RL algorithms (see \cite{Kain:2020vjs} for a discussion), adding robustness to the training in case of data scarcity. While showing good potential, the type of data-augmentation technique used in this study is quite simplistic, and can be improved and generalised in future works.

The study in Section~\ref{sec:online} illustrate how RL can increase the level of automation of DQM beyond a simple classification task, by including the cost in terms of person power in the overall optimisation goal. Our example shows a setup in which the algorithm can make automatic and accurate decisions at the maximum possible rate (every time the monitoring histograms are produced), while requiring intervention from the shifter only when the algorithm is not confident enough on its decision, i.e. before it is sufficiently trained or when the operating conditions have changed. The multi-agent setup could be expanded in the future, to handle a scenario in which many of the different interconnected actions required for the control of the data collection in an experiment would be automated and optimised jointly. This would require an investigation of how to construct a reward scheme from an evaluation of the relative costs of the actions and outcomes. In that regard, while the current paper does not discuss different possible types of reward scheme, one could imagine an alternative approach in which the reward is proportional to the final amount of data collected in nominal conditions, modified by some penalisation related to an estimated cost of the different possible actions.

The studies with noisy labels and data augmentation have only been performed in the offline case. Applying noise to the labels in the online case is expected to give similar results to the experiment in Section~\ref{sec:exp_superhuman}. The usage of data augmentation techniques is however more challenging than in the offline case, as the time dependence between successive states needs to be taken into account. We leave this for future work.

We note that the computational time requirements for algorithm training, while depending on many factors such as the hardware architecture (see Section~\ref{sec:resources}) and the data augmentation scheme, are low in both regimes. In the current studies, each training iteration was always found to be below 1 s, except for the case with the data augmentation, where it takes about 2 s. Assuming the DQM histograms are typically generated at fixed intervals of at least 5 min, training time is not expected to be a showstopper for the proposed approach.

The studies presented in this paper are not specific to a particular experimental setup, and the approach is therefore general, with potential application in DQM problems within or outside of particle physics.

\subsection{Limitations}
\label{sec:limitations}

The main challenge that we foresee when bringing the presented approach to reality concerns the robustness of the training against shifters' biases. These biases may arise from a mismatch between the average knowledge of the shifters and that of the small set of detector experts in the experiment. One way to address this is to build ``reference'' histograms that reflect what a ``nominal'' histogram would look like based on the shifters' decision, and then have these references checked with a certain frequency by detector experts so that biases in the shifters' decisions can be addressed. The construction of such a reference is already part of our data augmentation approach.

\section{Conclusions}
\label{sec:conclusions}

This letter proposes for the first time the application of RLHL techniques to perform DQM in large particle physics experiments, aiming to provide a proof-of-concept for an automated monitoring system that adapts to changing conditions and trains continuously with human shifters via mutual feedback. We describe the main differences between two types of operational regimes, online and offline, and prove that our proposed algorithms perform well for both of them based on a simplified synthetic dataset. We propose a data augmentation technique to improve the data-sample efficiency and show its effectiveness. Our studies show that the proposed algorithm automatically adapts to changing conditions, can improve the accuracy on the detection of anomalies by reducing random unbiased fluctuations in the shifter's decisions and can balance out accuracy with the need for shifter intervention in the online regime. Besides providing the foundational ground for applications in particle physics experiments, this work opens the door for future research on complex human--machine interactions in control rooms, with the potential of greatly optimising the human costs of operation while achieving super-human performance.

\section*{Acknowledgements}

The work of OJP has been sponsored and supported by the CERN openlab (openlab.cern). TL has been partially supported by the SNSF Advanced Grant {\it TMAG-2\_209263}. This publication is part of SK's project \textit{Uncovering the lepton generation gap} (with project number VI.Veni.202.004 of the research programme Veni which is partly financed by the Dutch Research Council (NWO).

\section*{Authorship declaration}

OJP contributed to the original idea, participated in the design of the algorithms, proposed the specific RL architecture, co-proposed the multi-agent setup, carried out most of the software development in collaboration with JGP and LDPP, participated in the design of the experiments and contributed to the writing of sections 3, 4, 5 and 6. JGP started and coordinated the project, directly supervised the work of OJP and LDPP, contributed to the original idea, designed the experiments, participated in the software development and reviewing, proposed and developed the data-augmentation scheme, designed the structure of the paper and contributed to its general writing and reviewing. LDPP participated in the design phase of the algorithm and in the software development, carrying out the initial version of several of the offline regime studies as part of his M.Sc. thesis project, under the supervision of JGP and in collaboration with OJP. MJ contributed to the drafting and further versions of sections 3 and 4, wrote the algorithm description in the appendix, reviewed and modified in detail sections 5 and 6 as well as appendix A.2, drafted the first version of a formalization of the problem statement, participated in the discussion, proposal and reviewing of reward schemes, contributed to the publication strategy. SK provided expert knowledge throughout the project, participating in the conceptual design of the algorithm’s operational scheme and the configuration of the online and offline setups. TL co-proposed the multi-agent architecture, proposed the reward scheme that led to the final algorithm, contributed to the investigation of the superhuman performance in the presence of noise, reviewed the PPO policy software, wrote the final version of the abstract as well as sections 3 and 4 of the paper and contributed to the writing of the rest of the paper. NS contributed to the original idea, participated in the design phase of the algorithm, co-suggested the use of reinforcement learning together with JGP, co-proposed testing super-human conditions, helped in identifying challenges and experiments, supervised OJP.

\newpage

\bibliographystyle{naturemag}

\newpage

\appendix

\section{Appendix}
\label{sec:appendix}

\subsection{Synthetic dataset}
\label{sec:dataset}

The synthetic data points used in the different experiments in this paper take the form of one-dimensional histograms of 100 bins each. The histograms are generated assuming a certain underlying distribution, $f$, that we model with either Gaussian, $f_g$, or exponential, $f_e$, functions, parameterised as follows:
\begin{equation}
f_g(x)=c_{g1}\cdot\exp{\left(\frac{x-c_{g2}}{2c_{g3}^2}\right)},
\end{equation}
\begin{equation}
f_e(x)=c_{e1}\cdot\exp{(c_{e2}x)},
\end{equation}
with $c_{g1}$, $c_{g2}$, $c_{g3}$, $c_{e1}$ and $c_{e2}$ being constant parameters. We fix $c_{g1}=10$ and $c_{e1}=1$ in all the cases, while the other parameters are set as explained below. The generation of each histogram is done bin by bin, using the function's value at the bin's center and adding a random Gaussian fluctuation to its value aimed at emulating stochastic effects in real data, such as statistical fluctuations related to the number of counts. For the experiments in this paper, the Gaussian noise uses a constant width of 0.01 for all the bins. Either the Gaussian or exponential distributions can represent a nominal or an anomalous state within a given experiment. If they are representing a nominal state, the distribution shape parameters will remain fixed. If they are representing an anomalous state, some of their values will be sampled from a certain range. The values chosen for all the experiments described in the paper are written in Table~\ref{tab:distribution_params_setup1}, Table~\ref{tab:distribution_params_setup2} and Table~\ref{tab:distribution_params_setup3}.

\begin{table}
  \caption{Distributions used for the experiments in Section~\ref{sec:exp_changing_conditions}and Section~\ref{sec:online}.}
  \label{tab:distribution_params_setup1}
  \centering
  \begin{tabular}{ll|c|c}
    \toprule
    \multicolumn{2}{c|}{Episodes} & $[0,999]$ & $[1000,1999]$ \\
    \toprule
    \multirow{2}{*}{Nominal distribution} & Function & $f_e$ & $f_g$ \\
    & Parameters & $c_{2e}=-0.0002$ & $c_{g2}=10$, $c_{g3}=30$ \\
    \midrule
    \multirow{2}{*}{Anomalous distribution} & Function & \multicolumn{2}{c}{$f_g$} \\
    & Parameters & \multicolumn{2}{c}{$c_{g2}\in[10,80]$, $c_{g3}=30$} \\
    \bottomrule
  \end{tabular}
\end{table}

\begin{table}
  \caption{Distributions used for the experiment in Section~\ref{sec:exp_superhuman} and Section~\ref{sec:exp_superhuman_with_human}.}
  \label{tab:distribution_params_setup2}
  \centering
  \begin{tabular}{ll|c}
    \toprule
    \multicolumn{2}{c|}{Episodes} & $[0,1999]$ \\
    \toprule
    \multirow{2}{*}{Nominal distribution} & Function & $f_g$ \\
    & Parameters & $c_{g2}=10$, $c_{g3}=30$ \\
    \midrule
    \multirow{2}{*}{Anomalous distribution} & Function & $f_e$ \\
    & Parameters & $c_{e2}\in[-0.005,0.005]$ \\
    \bottomrule
  \end{tabular}
\end{table}

\begin{table}
  \caption{Distributions used for the experiment in Section~\ref{sec:exp_data_augmentation}.}
  \label{tab:distribution_params_setup3}
  \centering
  \begin{tabular}{ll|c|c|c}
    \toprule
    \multicolumn{2}{c|}{Episodes} & $[0,49]$ & $[50,99]$ & $[100,149]$ \\
    \toprule
    \multirow{2}{*}{Nominal distribution} & Function & \multicolumn{2}{c|}{$f_g$} & $f_e$ \\
    & Parameters & \multicolumn{2}{c|}{$c_{g2}=10$, $c_{g3}=30$} & $c_{2e}=0.0002$ \\
    \midrule
    \multirow{2}{*}{Anomalous distribution} & Function & $f_e$ & \multicolumn{2}{c}{$f_g$} \\
    & Parameters & $c_{e2}\in[-0.005,0.005]$ & \multicolumn{2}{c}{$c_{g2}\in[0,100]$, $c_{g3}=40$} \\
    \bottomrule
  \end{tabular}
\end{table}

\subsection{Trust function}
\label{sec:trust_function}
In this appendix, we describe the function $f$ which assigns to $p_{\text{correct}}^{\text{proxy}}$, as introduced in Section~\ref{sec:exp_superhuman_with_human}, the probability with which a shifter having access to the output of the algorithm will randomly decide to replace their decision with that of the algorithm.

In general, it seems reasonable to assume that this function is monotonically increasing. The shifters may not use the decisions of the algorithm at all if $p_{\text{correct}}^{\text{proxy}}$ is very low. Furthermore, in order to allow for shifters to not strictly follow the algorithm even if the latter is very confident in its classification (that is, when $p_{\text{correct}}^{\text{proxy}}\approx 1$), we saturate the probability of having a shifter replace its decision by that of the algorithm at some $s<1$. For our demonstration, we choose the function $f$ as follows:
\begin{equation}
    f(p_{\text{correct}}^{\text{proxy}})=\frac{s}{1+\exp(-(p_{\text{correct}}^{\text{proxy}}-p_0)/w)},
\end{equation}
where $s$ is the value to which the probability saturates, $p_0$ is the point at which the majority of shifters adapt the decision of the algorithm, and $w$ is a parameter that controls the width of the region in which the probability increases in a quasi-linear way. The heuristically chosen values for those parameters are written up in Table~\ref{tab:trust_function}, and the resulting function is shown in Figure~\ref{fig:trust_fuction}.

\begin{table}
  \caption{Parameters used in the trust function.}
  \label{tab:trust_function}
  \centering
  \begin{tabular}{ll}
    \toprule
    Name     & Chosen value \\
    \midrule
    $s$ & $0.95$     \\
    $p_0$ & $0.7$     \\
    $w$ & $0.02$     \\
    \bottomrule
  \end{tabular}
\end{table}

\begin{figure}
  \centering
  \includegraphics[width=0.6\linewidth]{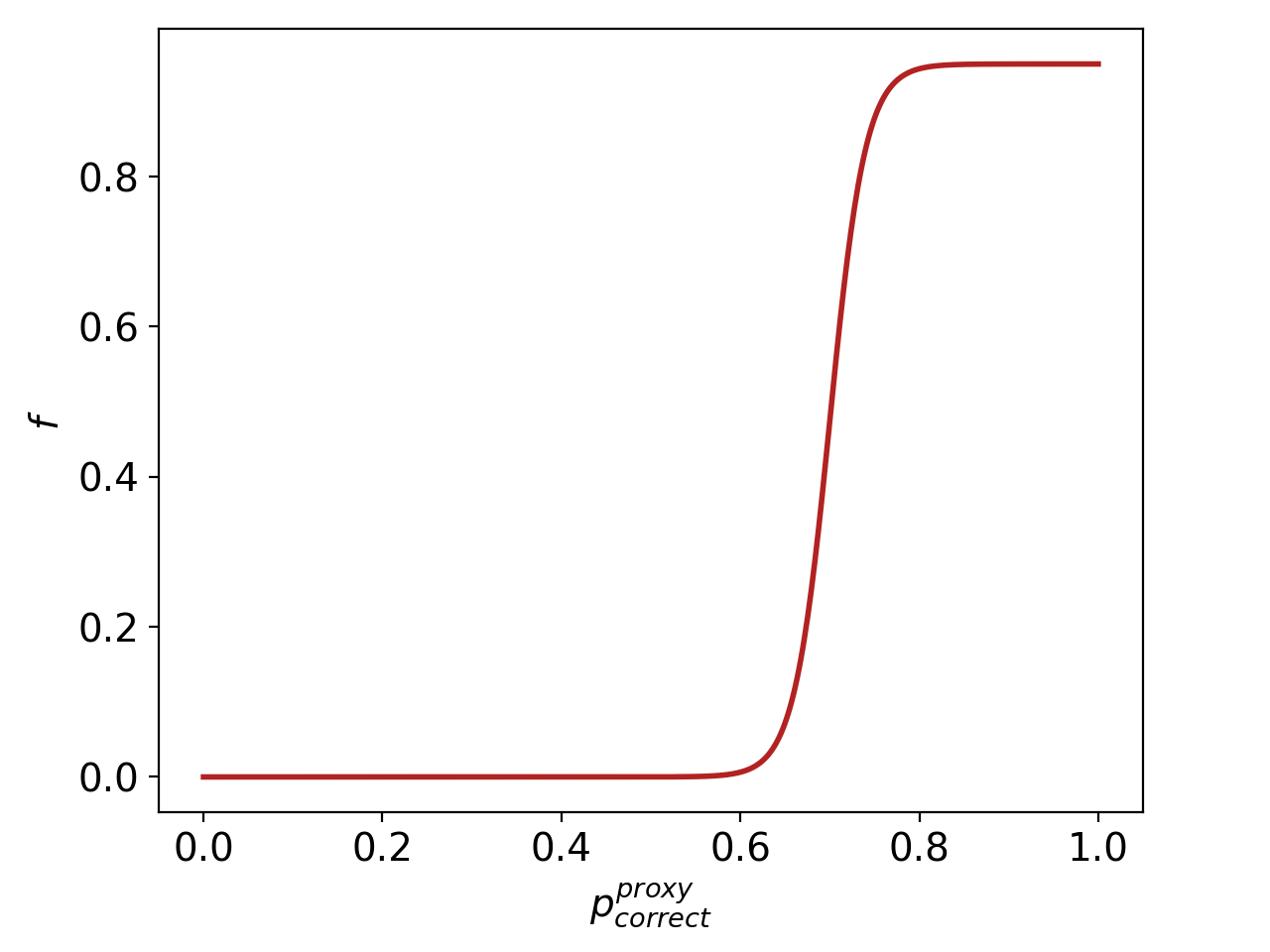}
  \caption{Function used to model the shifter's probability to trust the algorithm.}
  \label{fig:trust_fuction}
\end{figure}

\subsection{Data augmentation algorithm}
\label{sec:augmentation}

This section describes the process to generate an artificial histogram, represented by the input vector $z_{bin}$ and the corresponding artificial label, nominal or anomalous. The procedure includes the construction and continuous update of a reference histogram, represented by a vector of bin centers, $\mu_{bin}$, and a vector of bin uncertainties, $\sigma_{bin}$. This histogram takes information from the input vector, $x_{bin}$, of the original histograms in the training dataset that are labelled as nominal. The update of the reference histogram over time is based on the Exponentially Weighted Moving Average (EWMA) approach~\cite{264b1310-db27-36ce-b68f-3e439b03c2b5}, to allow the adaptation to changing conditions over time.

The structure of the algorithm is described in Algorithm~\ref{alg:data_augmentation} and its hyperparameters are described in Table~\ref{tab:data_augmentation_params}. Two different parameters, $\alpha_{\mu}$ and $\alpha_{\sigma}$, are used to control separately the time evolution of the bin centers and uncertainties of the reference histogram. Adjusting their values allows to change the relative impact on the reference histogram of more recent nominal histograms in the training dataset compared to older ones. The data augmentation process is started only after a certain number of training steps, to allow time for the reference template to sufficiently stabilise. All the parameters have been moderately tuned by hand for the experiment described in Section~\ref{sec:exp_data_augmentation}. An in-detail hyperparameter optimisation is out of the scope of this paper.

\begin{table}
  \caption{Parameters used in the data augmentation algorithm.}
  \label{tab:data_augmentation_params}
  \centering
  \begin{tabular}{lll}
    \toprule
    Name     & Description     & Chosen value \\
    \midrule
    $\alpha_{\mu}$ & EWMA parameter for bin centers  & $0.99$     \\
    $\alpha_{\sigma}$ & EWMA parameter for bin uncertainties  & $0.5$     \\
    $AugmStartPoint$ & Number of the first training step with augmentation  & $30$     \\
    $AugmFactor$ & Number of augmented states for each original one  & $499$     \\
    $AugmAnomProb$ & Probability that an augmented state is an anomaly  & $50\%$     \\
    $MinNumVars$ & Minimum number of variations  & $1$     \\
    $MaxNumVars$ & Maximum number of variations  & $50$     \\
    $MinVarSize$ & Minimum number of bins in each variation & $1$     \\
    $MaxVarSize$ & Maximum number of bins in each variation & $100$     \\
    $MinNumStdDev$ & Minimum number of standard deviations & $5$     \\
    $MaxNumStdDev$ & Maximum number of standard deviations & $20$     \\
    \bottomrule
  \end{tabular}
\end{table}

\begin{algorithm}[H]
\caption{Training with data augmentation}\label{alg:data_augmentation}
\begin{algorithmic}
\State Initialise $\mu_{bin}$ to 0
\State Initialise $\sigma_{bin}$ to 1
\For{each histogram in training dataset}
    \State Do one training step using $x_{bin}$ and its label
    \If{the current histogram is nominal}
    \Comment{Update the reference histogram}
    \State $\mu_{bin}\gets\alpha_{\mu}\cdot x_{bin}+(1-\alpha_{\mu})\cdot\mu_{bin}$
    \State $\sigma_{bin}\gets\sqrt{\alpha_{\sigma}\cdot(x_{bin}-\mu_{bin})^2+(1-\alpha_{\sigma})\cdot\sigma_{bin}^2}$
    \EndIf
    \If{the training step is $\geq AugmStartPoint$}
    \Comment{Start the data-augmentation process}
    \For{$AugmF actor$ iterations}
    \State Generate an augmented histogram $z_{bin}\sim N(\mu_{bin},\sigma_{bin})$
    \State Sample whether the state will be an anomaly, with probability $AugmAnomProb$
    \State Set the histogram label according to whether it is anomalous or not
    \If{the state is an anomaly}
    \For{a number of iterations in $[MinNumVars,MaxNumVars]$}
    \State Sample an integer variation size in $[MinVarSize,MaxVarSize]$
    \State Sample the position of the first (left-most) bin in the variation
    \State Sample a real scaling factor, $SF$, in $[MinNumStdDev,MaxNumStdDev]$
    \State Randomly make the $SF$ negative with a 50\% probability
    \State Modify every bin content in the variation, $bin'$, as $z_{bin'}\gets z_{bin'}+SF\cdot\sigma_{bin'}$
    \EndFor
    \EndIf
    \State Do one training step using $z_{bin}$ and its label
    \EndFor
    \EndIf
\EndFor
\end{algorithmic}
\end{algorithm}

\subsection{Review of Reinforcement Learning and Proximal Policy Optimization}\label{sect:appendix-PPO}
In this section we review the basic theory of Proximal Policy Optimization, starting with an elementary recall of Reinforcement Learning, going over basic algorithms aiming to find optimal policies, so-called Policy Gradient Methods, and finally introducing Proximal Policy Optimization.

\subsubsection{Formal setting for reinforcement learning}\label{sect:formal-setting}
Recall that a reinforcement learning agent is tasked with maximizing a reward in a Markov decision process. A \emph{Markov decision process} is a tuple $(\kS,\kA,\rho,\P_0)$ consisting of a measurable space $\kS$, called the set of \emph{states}, as well as a measurable space $\kA$, called the set of \emph{actions}. $\rho$ is a Markovian transition kernel from $\kS\times\kA$ to $\R\times\kS$, called the \emph{environment dynamics}, such that for every $s\in\kS$ and $a\in\kA$, $\rho(\cdot,\cdot\mid s, a)$ is a probability measure on $\R\times\kS$ defining the joint probability distribution of the reward and the next state if one starts in state $s$ and performs action $a$. $\P_0$ is a probability distribution on $\kS$. We start at a state $S_0$ which is $\P_0$-distributed. 

At every time $t\in\N$, the agent performs an action according to a \emph{policy} $\pi_t$. The policy $\pi_t$ is a Markovian transition kernel from $\kS$ to $\kA$. For every $s\in\kS$, $\pi_t(\cdot\mid s)$ defines the probability distribution of which action to take if at time $t$ one is in state $s$.

For $t\in\Z_{\ge 0}$, we define recursively random variables $A_t, R_t, S_{t+1}$ which satisfy that the conditional distribution of $A_t$ given $S_t$ is $\pi_t(\cdot\mid S_t)$ and that
\begin{equation}
(R_t, S_{t+1})\mid S_t, A_t\sim \rho(\cdot, \cdot\mid S_t, A_t).
\end{equation}
\begin{remark}[Formal Definition of the $A_t, R_t, S_{t+1}$]
    To give a formal Definition, we may proceed as follows: We fix first a probability space with two families of random variables $(X_{t,s})_{t\in\Z_{\ge 0}, s\in\kS}$ and $(Y_{t,s,a})_{t\in\Z_{\ge 0}, s\in\kS, a\in\kA}$ which we assume to be jointly independent, and which we assume to satisfy $X_{t,s}\sim \pi_t(\cdot\mid s)$ as well as $Y_{t,s,a}\sim\rho(\cdot, \cdot\mid s,a)$ for all $t\in\Z_{\ge 0}, a\in\kA, s\in\kS$. (The existence of such a space is guaranteed for instance by the Andersen-Jessen Theorem.) We then set, recursively,
    \begin{equation}
        A_t\define X_{t, S_t}\qquad (R_t, S_{t+1})\define Y_{t, S_t, A_t}
    \end{equation}
    for $t\in\Z_{\ge 0}$.
\end{remark}
The random variable $A_t$ is the action one takes at time $t$, while $S_t$ is the state at time $t$ and $R_t$ is the reward obtained at time $t$.

Note that all of these random variables depend on the choice of the policies $\pi_t$. The goal of the agent is now to choose policies which maximize the cumulative expected reward
\begin{equation}
    \sum_{t\in\N} \gamma^t \E(R_t),
\end{equation}
where $\gamma\in\R_{\ge 0}$ is a so-called \emph{discounting factor}.

\begin{definition}[State value function]
    The \emph{state value function at time $t$} assigns to a state $s\in\kS$ and time $t\in\N$ the expected future cumulative if we are in this state:
    \begin{equation}
        V(t, s)\define \sum_{t\in\N} \gamma^t\E(R_t\mid S_t = s).
    \end{equation}
    (In a finite horizon setting we may only sum over all $t$ up to some time $T\in\N$.)
\end{definition}

Summarizing, there are two main parts of a reinforcement learning algorithm:
\begin{description}
    \item[Policy] The policies $\pi_t$ defining what action to take at time $t$ for each state;
    \item[State value function] The state value function $V$ giving the expected future cumulative reward at time $t$ for each state.
\end{description}

\subsubsection{Policy Gradient Methods}\label{sect:PGM}
To illustrate how optimization of the policies proceeds, we first describe so-called \emph{Policy Gradient Methods (PGMs)}. PGMs have a parametrized family of policies $(\pi^\theta_t)_{\theta\in\Theta}$, where $\Theta$ is some index set. The goal of PGMs is to update the parameters $\theta$ in the direction of the gradient of the expected return, thus optimizing the parameters of the policy using a form of gradient ascent. More formally, note that $R_t$ from Section \ref{sect:formal-setting} can be considered a function of $\theta$. PGM algorithms such as REINFORCE, introduced by \cite{Williams:92}, iterate by performing policy ascent. Starting at some $\theta_0\in\Theta$, one then iterates as follows: 
\begin{equation}
    \theta_{k+1}=\theta_k + \alpha\nabla_\theta \E(R_t),
\end{equation}
where $\alpha\in\R_{\ge 0}$ is the \emph{learning rate}.

\begin{remark}[Summing over time]\label{rem:summing}
    In the current formulation, we are optimizing the policy of a single time step $t$ only. One can formulate the same idea by parameter-sharing across all policies $(\pi_t)_{t\in\Z_{\ge0}}$ and taking the derivative of a discounted reward $\sum_{t\in\Z_{\ge 0}} \gamma^t R_t$ or a finite-horizon reward $\sum_{t=0}^T R_t$ instead. For ease of exposition, we limit ourselves to one time step but this Remark is valid throughout the rest of the section.
\end{remark}

In order to compute $\nabla_\theta\E(R_t)$, REINFORCE employs the following Lemma.

\begin{lemma}[Log-derivative trick]\label{lem:log-derivative trick}
    Let $(p^\theta)_{\theta\in\Theta}$ be a family of probability densities with respect to a $\sigma$-finite reference measure $\mu$ on some measurable space $\Omega$ indexed by some open set $\Theta\subset\R^n, n\in\N$, and denote by $\P^\theta$ the corresponding probability measures on $\Omega$. Let $f:\Omega\to\R$ be a measurable function. Then, if the family of functions $x\mapsto p^\theta(x) f(x)$ is sufficiently regular for the Leibniz rule to hold,\footnote{We omit a further discussion on what conditions are sufficient for this statement.} we have
    \begin{equation}
        \nabla_\theta \E^{\P^\theta}(f) = \E^{\P^\theta} (\nabla_\theta \ln(p^\theta) f),
    \end{equation}
    where $\E^{\P^\theta}$ denotes the expectation with respect to $\P^\theta$.
\end{lemma}
\begin{proof}
    Using the Leibniz rule to exchange differentiation and integration, we get
    \begin{equation}\begin{split}
        \nabla_\theta\E^{\P^\theta}(f) &= \nabla_\theta\int_{\Omega} p^\theta(x) f(x)\,\mathrm d\mu(x) \\
        &= \int_{\Omega} \nabla_\theta p^\theta(x) f(x)\,\mathrm d\mu(x) \\
        &= \int_{\Omega} \frac{\nabla_\theta p^\theta(x)}{p^\theta(x)} f(x) p^\theta(x)\,\mathrm d\mu(x) \\
        &= \int_{\Omega}\nabla_\theta\ln(p^\theta(x)) \,\mathrm d\P^\theta(x) \\
        &=\E^{\P^\theta} (\nabla_\theta \ln(p^\theta) f). \qedhere
    \end{split}\end{equation}
\end{proof}
To derive the REINFORCE algorithm, we use Lemma \ref{lem:log-derivative trick} applied to the family $\P^\theta$ of joint distributions over the $R_t, S_t, A_t$ (this distribution depends on $\theta$ through the choice of policies $\pi_t^\theta$) in order to obtain, under the assumption that $\P^\theta$ have densities $p^\theta$ satisfying the conditions of Lemma \ref{lem:log-derivative trick},  that 
\begin{equation}
    \nabla_\theta\E\left(R_t\right) = \E\left(\nabla_\theta \ln \pi^\theta_t(A_t \mid S_t) R_t\right).
\end{equation}

\subsubsection{Proximal Policy Optimization}\label{sect:PPO-final}
Having introduced basic PGMs in Section \ref{sect:PGM}, we now describe the approach used in our article, Proximal Policy Optimization (PPO). A big problem with PGMs such as REINFORCE is that the updates introduced through gradient ascent may be numerically instable. PPO, introduced in \cite{schulman2017proximal}, aims to solve this problem by recasting the optimization problem.

For $t\in\Z_{\ge 0}$, we define the \emph{advantage function at time $t$} as the function $\mathfrak A_t$ mapping each $s\in\kS$ and $a\in\kA$ to
\begin{equation}\label{eq:advantage-function}
    \mathfrak A_t(s,a)\define \E(R_t\mid S_t = s \land A_t = a) - \E(R_t\mid S_t=s).
\end{equation}
The advantage function gives for a state $s$ and action $a$ the difference in expected reward if we perform action $a$ instead of performing a random action drawn from the policy $\pi_t(\cdot\mid S_t)$. Since the advantage function depends on the policies, we will write $\mathfrak A_t^{\pi}$ for the advantage policy associated of policy $\pi$ where the choice of the policy is not evident from the context. We note that, for ease of exposition, we only give the Definition for a single time step. In a full implementation, one may sum up to a finite horizon or approximate the full discounted future reward, see Remark \ref{rem:summing}.

Proximal policy optimization, just like REINFORCE, iteratively updates a set of parameters $\theta\in\Theta$ which parameterize a family of policies $\pi_t^\theta$. We start with an initial set of parameters $\theta_0$ and then update through the following step.
\begin{equation}
    \theta_{k+1} \in\operatorname{arg max}_{\theta\in\Theta} \E(L^{\text{clip}}(S_t, A_t, \theta_k,\theta)),
\end{equation}
where
\begin{equation}
    L^{\text{clip}}(s, a, \theta_k,\theta)\define \min\left(r(\theta,\theta_k)\mathfrak A^{\pi_t^{\theta_k}}(s, a), \min\left(1+\varepsilon, \max\left(1-\varepsilon, r(\theta,\theta_k, a, s)\right)\right) \mathfrak A^{\pi_t^{\theta_k}}(s,a)\right)
\end{equation}

with
\begin{equation}
    r(\theta,\theta_k) \define \frac{\mathrm d\pi_t^{\theta}(\cdot\mid s)}{\mathrm d\pi_t^{\theta_k}(\cdot\mid s)}(a).
\end{equation}
We assume here implicitly that $\pi_t^{\theta}(\cdot\mid s)$ is absolutely continuous with respect to $\pi_t^{\theta_k}(\cdot\mid s)$. 

A main question for proximal policy optimization is how to access the advantage function. This is because the expectations in \ref{eq:advantage-function} are not directly accessible. Instead, we estimate it using a \emph{value network}, whose task it is to approximate the state value function. The policy updates then use values provided by the value network. In algorithm \ref{alg:PPO-pseudo-code}, we denote the value network with parameters $\phi$ by $V^\phi_t$.

Therefore, PPO is an \emph{actor-critic} approach. In the actor-critic approach, two functions get trained concurrently. The first function, the so-called \emph{critic}, approximates the state value function. The second function, the so-called \emph{actor}, is the policy of the agent aiming to maximize the cumulative discounted expected reward. The actor uses values generated by the critic in order to update its parameters.

\begin{algorithm}
\caption{PPO-Clip natural language pseudo-code, minimally adapted from \cite{PPO-OpenAI-SpinningUp}}\label{alg:PPO-pseudo-code}
\begin{algorithmic}
\State Input: Initial policy parameters $\theta_0$, initial value function parameters $\phi_0$.
\For{$k=0,1,2,\dots$}
\State Generate a sequence of realizations of states, actions and rewards according to the policies $\pi_0^{\theta_0},\pi_1^{\theta_0},\dots$.
\State Compute estimates of the advantage function based on the current value network $V^{\phi_k}_t$.
\State Update the policy by maximizing the PPO-Clip objective $L^{\text{Clip}}$, typically with stochastic gradient ascent.
\State Update value network parameters $\phi_k$ to those $\phi_{k+1}$ which minimize the empirically observed expectation of $(V_t^\phi(S_t)-R_t)^2$, typically via gradient descent.
\EndFor
\end{algorithmic}
\end{algorithm}

\clearpage

\end{document}